\documentclass[english,11pt]{article}

\usepackage{amsmath, amsthm, amssymb, amsfonts, graphicx}
\usepackage{tikz}
\usepackage{enumerate}
\usepackage{cite}
\usepackage{clrscode}

\newtheorem{theorem}{Theorem}
\newtheorem{definition}[theorem]{Definition}

\newtheorem*{theorem*}{Theorem}
\newtheorem*{remark*}{Remark}

{\makeatletter
 \gdef\xxxmark{%
   \expandafter\ifx\csname @mpargs\endcsname\relax 
     \expandafter\ifx\csname @captype\endcsname\relax 
       \marginpar{xxx}
     \else
       xxx 
     \fi
   \else
     xxx 
   \fi}
 \gdef\xxx{\@ifnextchar[\xxx@lab\xxx@nolab}
 \long\gdef\xxx@lab[#1]#2{\textbf{[\xxxmark #2 ---{\sc #1}]}}
 \long\gdef\xxx@nolab#1{\textbf{[\xxxmark #1]}}
}

\usepgflibrary{shapes}
\newtheorem{thm}{Theorem}

\bigskip

\begin{document}

\title{Improved Connectivity Condition for Byzantine Fault Tolerance}

\author{Adam Hesterberg \thanks{MIT Math Department, 77 Massachusetts Avenue, Cambridge, MA, 02139, USA. Email: \texttt{achester@math.mit.edu}. Phone: 1(609)616-2849}, Andrea Lincoln\thanks{MIT Computer Science and Artificial Intelligence Laboratory, 32 Vassar Street, Cambridge, MA 02139,
USA. Email:~\texttt{\{andreali,jaysonl\}@.mit.edu}. Phone: 1(650)465-0538}, \& Jayson Lynch \footnotemark[4]}

\date{}

\maketitle

\setcounter{page}{0}

\begin{abstract}
Given a network in which some pairs of nodes can communicate freely, and some subsets of the nodes could be faulty and colluding to disrupt communication, when can messages reliably be sent from one given node to another? We give a new characterization of when the agreement problem can be solved and provide an agreement algorithm which can reach agreement when the number of Byzantine nodes along each minimal vertex cut is bounded. Our new bound holds for a strict superset of cases than the  previously known bound. We show that the new bound is tight. 
Furthermore, we show that this algorithm does not require the processes to know the graph structure, as the previously known algorithm did. Finally, we explore some of the situations in which we can reach agreement if we assume that individual nodes or entire subgraphs are trustworthy.
\end{abstract}

\clearpage

\section{Introduction}

In this paper we investigate the Agreement problem in synchronous systems with Byzantine nodes: given a network in which some pairs of nodes can communicate freely, with some subsets of the nodes possibly faulty, when can messages reliably be sent from one node to another? We'll give characterizations of when this problem can be solved. First, we give an alternate agreement algorithm which can reach agreement when the number of Byzantine nodes in each minimal vertex cut is restricted. This is a strictly weaker requirement than was previously known, and we show that it's tight in that if there are any more Byzantine nodes than it allows, then no solution exists. Next, we show that this algorithm does not require the processes to know the graph structure, as the previously known algorithm did. Finally, we explore some of the situations in which we can reach agreement if we assume that individual nodes or entire subgraphs are trustworthy.

In the real world, there are often reasonable assumptions that can be made about the distribution of faulty nodes in a system. By giving a broader characterization of when agreement with Byzantine faults can be solved, we enable more potential application. For example, one might have a very large grid network which is only 4-connected. However, if the designers do not expect that Byzantine nodes can be adversarially placed, then they may be confident that the Byzantine faults will likely be spaced out and thus they can tolerate far more than four faults in the entire system, where the previously known algorithm could only tolerate one. If one is interested in faults occurring at random nodes rather than in a purely Byzantine manner we may have a far more robust system. Randomly distributed Byzantine faults have been studied by Blough and Pelc \cite{RandomDist}.  Another work, by Reischuk, looks at a similarly styled restriction where only a certain fraction of the nodes participating in a phase behave in a faulty manner \cite{Reischuk198523}. In terms of trusted nodes, it is perfectly reasonable for an organization to have part of its network maintained at a higher level of security, say with hardware-verification, to the extent that it is either unreasonable to believe that part of the network has failed or that processor being compromised makes the rest of the work a moot point. Also, peer-to-peer networks such as Tor are intended to be robust against some of the ``peers'' being compromised, and it is more likely that, say, all the US nodes or all the Chinese ones would be compromised simultaneously than ones from a mix of countries. Thus there are practical concerns that can be addressed by these extensions.

This work primarily builds off of the initial work (in \cite{LSP80} and \cite{LSP82}) by Lamport, Shostak, and Pease giving algorithms for solving Agreement with Byzantine Faults as well as necessary and sufficient conditions for solving the problem in complete graphs. We also note Dolev who generalized this to non-complete graphs in \cite{Dolev82} and the work of Dolev, Strong, Castro, and Liskov who haver given additional constraints to the problem or graph structure pushing allowing for more potential practical application. \cite{DS83} \cite{CL99}

There are three major sections to this paper. First, we give motivating examples where distributed systems with certain properties appear to be able to reach agreement, even though they do not fulfill the general criteria given by LSP. Next, we prove our central theorem, giving an exact condition when agreement can be reached in the presence of Byzantine faults. The third section goes on to consider some of the implications of this theorem, and also extends the work by applying the results to cases where processors do not have less knowledge of the graph than would normally be needed to solve agreement and gaining stronger results in the presence of trusted nodes.

\section{Motivating Examples}
Byzantine Agreement has tight bounds on the number of nodes required ($n > 3f$) and the connectivity required ($c > 2f$) for agreement. So, in order to motivate our results we need to show examples where the original Byzantine Agreement conditions are too stringent. 

In all these cases we assume that nodes have UIDs and that when a node receives a message it can determine from which neighbor it received the message. However, there are no signatures. Thus, if A sends to C through B: C can't determine if A really sent the original message, but can determine that B sent the most recent one.  

When we are analyzing the case where we limit the distribution of nodes: we can sometimes perform better than $n>3f$ and $c>2f$. After all, with limitations on the placement of Byzantine nodes we limit the cases of failure we need to consider. At the most extreme, we could limit all f faults to f nodes, allowing the nodes that are assumed non-faulty to communicate directly and ignore the potentially faulty nodes. 

When analyzing the case where we have an unknown graph we want to show that the existing algorithm does not solve the problem even when $c>2f$ and $n>3f$. This will motivate our conclusion that with an altered algorithm for node to node communication, we can solve this problem in all cases where $c>2f$ and $n>3f$. 


\subsection{Majority Non-faulty Nodes Per Cut}
Consider the case of $K_{4,3}$. This graph is 3-connected, so previous results about Byzantine Agreement tell us that we can't solve agreement for two faults. 
However, if we require that every minimal vertex cut set has a majority of non-faulty nodes, we can solve agreement for two faults. 

The four nodes forming one of the partite sets form a minimal vertex cut set and thus can have at most one fault. The three nodes forming the other partite set form a minimal vertex cut set and thus can have at most one fault. Thus, if we have two faults one must be in each partite set. This gives us a lot of power. The non-faulty nodes in opposite partite sets can communicate directly. The non-faulty nodes in the same partite set can communicate by taking the majority vote of the forwarded message from their direct neighbors: since there is at most one fault out of 3 or 4, non-faulty nodes will be the majority. Finally, the original agreement algorithm for complete graphs finishes the problem, because $7 \geq 3*2+1 = 7$. 

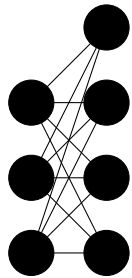
\begin{tikzpicture}
\draw[fill=black] (0,0) circle(0.3);
\draw[fill=black] (0,1) circle(0.3);
\draw[fill=black] (0,2) circle(0.3);
\draw[fill=black] (1,0) circle(0.3);
\draw[fill=black] (1,1) circle(0.3);
\draw[fill=black] (1,2) circle(0.3);
\draw[fill=black] (1,3) circle(0.3);
\draw (0,0) -- (1,0);
\draw (0,0) -- (1,1);
\draw (0,0) -- (1,2);
\draw (0,0) -- (1,3);
\draw (0,1) -- (1,0);
\draw (0,1) -- (1,1);
\draw (0,1) -- (1,2);
\draw (0,1) -- (1,3);
\draw (0,2) -- (1,0);
\draw (0,2) -- (1,1);
\draw (0,2) -- (1,2);
\draw (0,2) -- (1,3);

\end{tikzpicture}

This should motivate us to consider the condition of each cut being majority non-faulty nodes as providing interesting results.  In the results section it will be argued that with an altered node-to-node communication algorithm the condition of majority non-faulty nodes per vertex cut set and $n>3f$ together are sufficient to reach agreement. 

\subsection{Unknown Network Graph}
In the problem of the unknown network graph each node knows:
\begin{enumerate}
 \item the size of the network $n$,
\item the UIDs of all the nodes in the graph, and
\item the list of its own local neighbors.
\end{enumerate}

In the original problem what to do in the case of an unknown network graph is not even well defined. We select $2f+1$ disjoint paths; however, how to select $2f+1$ disjoint paths when the network is unknown is unclear. In some cases a Byzantine node can lie about network paths so that, naively, it looks like there are not $2f+1$ disjoint paths.

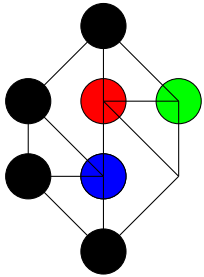
\begin{tikzpicture}
\draw[fill=black] (0,0) circle(0.3);
\draw[fill=black] (-1,1) circle(0.3);
\draw[fill=black] (0,1) circle(0.3);
\draw[fill=blue] (0,1) circle(0.3);
\draw[fill=black] (-1,2) circle(0.3);
\draw[fill=red] (0,2) circle(0.3);
\draw[fill=green] (1,2) circle(0.3);
\draw[fill=black] (0,3) circle(0.3);
\draw (0,3) -- (1,2);
\draw (0,3) -- (0,2);
\draw (0,3) -- (-1,2);
\draw (0,2) -- (0,1);
\draw (1,2) -- (1,1);
\draw (0,2) -- (1,2);
\draw (-1,2) -- (-1,1);
\draw (0,2) -- (1,1);
\draw (-1,2) -- (0,1);
\draw (-1,1) -- (0,1);
\draw (0,0) -- (0,1);
\draw (0,0) -- (-1,1);
\draw (0,0) -- (1,1);

\end{tikzpicture}

In this case, if the Byzantine node (red) claims not to be adjacent to the blue node, then it appears to the top node as if there are only two disjoint paths from A to B. With the original algorithm the selection process in this case for disjoint paths is undefined; if the path through the red node is chosen as one of the two, then the algorithm fails.

In results it will be argued that the altered node-to-node communication algorithm solves the issue of an unknown network graph. 

\subsection{Non-Faulty by Assumption}
To motivate the possible application of nodes being non-faulty by assumption: there may be some nodes that have a very very very low probability of failure, and can be trusted. Alternatively, for some network setups we can succeed only if certain nodes are good. For example, the central node of a star graph. 

Let us define a node that is “non-faulty by assumption” to be a node that can never be a Byzantine node, and is known to the whole graph to be a non-Byzantine node. 

To motivate the theory behind non-faulty by assumption nodes: consider the star graph $S_{n-1}$. This graph is one connected. However, if the central node is assumed non-faulty then agreement can be solved with up to $n-1$ faults by the central non-faulty node sending its value to all other nodes. This value is the one agreed to. Thus, we can far out-perform both the connectivity bound $c= 1 <2f+1$ and the number of nodes bound $n  <  3*(k-1)+1$.  With assumed non-faulty nodes we can perform better than before.

In results the effects of non-faulty nodes and cliques will be analyzed.

\section{Results}

We are specifying a new weaker cut condition, let us call it the \emph{weak cut property}. The goal of this definition is to capture what condition is needed on cuts if we have more information about the distribution of faulty nodes. 

\begin{definition}
The \emph{weak cut property} is the property that for any two sets $A$ and $B$ of vertices such that $A \cup B$ is a minimal vertex cut either 

\begin{itemize}
\item There is no valid distribution of faults such that the vertices in A can all simultaneously be faulty.
\item There is no valid distribution of faults such that the vertices in B can all simultaneously be faulty.
\end{itemize}
\end{definition}

To elucidate this definition we will show that the condition that the majority of every cut is non-faulty nodes implies the weak cut property. Let us partition a cut of size $c$ into two sets, $A$ and $B$. For any such partition, one of $A$ and $B$ will be of size $\geq c/2$. In a cut of size $c$ there are $< c/2$ faulty nodes. Thus, the largest set must have at least one non-faulty node in any valid distribution of faults. Thus, the condition that the majority of every cut is non-faulty implies the weak cut property. 

Note that the weak cut property covers strictly more cases than the previous condition ( if the smallest cut is of size $c$ then the number of faults must be less than $c/2$). Consider some cut, $C'$, in this graph. Let $|C'| = c'$, we know $c' \geq c$. When we partition this cut into $A$ and $B$ then one of them is of size $\geq c'/2$. Then the size of the larger set must be $\geq c'/2 \geq c/2 > f$. Thus, for any valid distribution of faults the larger half of a partition must have at least one non-faulty node. So, any graph that has the standard cut condition from previous Byzantine Fault Tolerance work also satisfies the weak cut property.

The overall structure of our proof is to show first that the weak cut property is necessary and second that it is sufficient. 
We will show that if the weak cut property does not hold, then any protocol will have a valid placement of Byzantine nodes that can cause non-agreement. 
Next we give a communication protocol that allows any non-faulty node $v$ can send a message to any non-faulty node $w$. 
Then, one can run the standard Byzantine agreement protocol that would be run on a fully connected graph, simply replacing every node to node communication protocol with the protocol we describe. 
If both $3f<n$ and every the weak cut property holds then the non-faulty nodes can agree. If either condition is not met then agreement is impossible. 


The intuition for the communication results is that minimum vertex-cut sets (shortened to cut from now on) are the nodes over which nodes on one side of the cut can communicate with those on the other side of the cut. Thus, if faults can be placed too freely no communication can occur. If faults are sufficiently limited in placement then communication is possible. 

\subsection{Agreement is impossible if the weak cut property is violated}

This first theorem is going to be the case where too much freedom is given to fault placement. $A$ and $B$ are two sets where $A \cup B$ forms a cut. If there exists a valid placement of faults such that all of $A$ can be faulty simultaneously. And, there exists a (possibly different) valid placement of faults such that all of $B$ can be faulty. Then, it is impossible to send messages across that cut. Thus, agreement is impossible.  

Intuitively, it is impossible to determine which of set $A$ and set $B$, which for a cut in the graph, to listen to if either could be entirely Byzantine. We prove this by a simple simulation argument.

\begin{thm}
If there exists a minimal vertex cut $C$ and two sets $A$ and $B$ of vertices such that $A \cup B = C$ and both 
\begin{itemize}
\item there is a valid distribution of faults such that the vertices in A can all simultaneously be faulty,
\item there is a valid distribution of faults such that the vertices in B can all simultaneously be faulty,
\end{itemize}
then  agreement is impossible.
\end{thm}

\begin{proof} Since $A \cup B$ is a cut, there are vertices $u$ and $v$ that are connected only through $A \cup B$. Let them both be non-faulty, and let their components of $G \setminus (A \cup B$ be $U$ and $V$, respectively. Consider initial setups $\alpha_{i,j}$ and $\beta_{i,j}$ in which all of $A$ and all of $B$, respectively, are faulty, all the vertices in $U$ start with the value $i$, and all the vertices in $V$ start with the value $j$. In $\alpha_{0,0}$, let the vertices in A simulate an execution in which every non-faulty vertex started with 1. All the nonfaulty vertices started with 0, so they must come to agreement on 0 anyway. But consider an execution starting from $\beta_{0,1}$ in which all the vertices in $B$ simulate their equivalents from the previous execution. From the perspective of the vertices in $U$, this is the same execution, so they must come to agreement on $0$, so the vertices in $V$ must also come to agreement on 0. Finally, consider an execution starting from $\alpha_{1,1}$ in which all the vertices in $A$ simulate their equivalents from the previous execution. From the perspective of the vertices in $V$, this is the same execution, so they must come to agreement on $0$, which is not any vertex’s starting value, violating validity.
\end{proof}

Bellow we state an equivalent theorem statement, which may be more intuitive, but less descriptive. 

\begin{thm}
If there is a cut that violates the weak cut property then agreement is impossible. 
\end{thm}

\subsection{Node to Node Communication Algorithm}

We will use the standard flooding algorithm, we will call it \emph{ FLOOD} for convenience. 

In \emph{FLOOD} let us say a node $v$ wants to send a message, $m$, to a node $w$. First $v$ sends the message $v_{uid}|m$, the UID of $v$ appended to the message $m$, to all neighbors of $v$. From this point on, all nodes will follow the following procedure:
\begin{itemize}
\item Take in messages. If you are $w$, do not go any further. 
\item Throw out any messages with non-sense headers
\item If the header already includes ones UID throw it out
\item Append your UID to the message and forward it to all your neighbors. 
\end{itemize}

After $w$ receives the first such message, it will wait for $|G|^2$ rounds and then take in the data. At this point $w$ will consider every possible valid set of failed nodes. For each possible set, $w$ will throw out all messages that include the UID of the failed node. If all the messages left after this have the same message body, then this is the message from $v$, use this message from $v$ to complete the higher level BFT algorithm. 

\subsection{If weak agreement holds then nodes can communicate}

When the conditions of that theorem hold---for instance, when an arbitrary half of the vertices of some cut can be faulty---we have no hope of achieving agreement. Therefore, we assume not. In that case, we show that every pair of non-faulty vertices can communicate reliably.

In this case, we are specifying that we can't have sets $A$ and $B$ with the properties above (that is they can't both have the potential to be fully faulty). To prove this, we look at the communication across a cut between two nodes. We show that there exists a self-consistent set of messages that are truthful, and then we show that given our assumption about the number of faulty nodes in the cut it is the maximal consistent set is truthful. Now we use this successful communication and induct to prove the full theorem.

\begin{thm}
If for any two sets $A$ and $B$ of vertices such that $A \cup B$ is a minimal vertex cut either 
\begin{itemize}
\item There is no valid distribution of faults such that the vertices in A can all simultaneously be faulty.
\item There is no valid distribution of faults such that the vertices in B can all simultaneously be faulty.
\end{itemize}

then every non-faulty vertex $v$ can communicate a message to any non-faulty vertex $w$ such that the message is guaranteed to be the same contents as the message sent by $v$.
\end{thm}

\begin{proof}
Let $u$ and $v$ be two non-faulty vertices, where $u$ wants to send a message $m$ to $v$. $u$ sends all its neighbors the message ``$u$ says $m$ to $v$’’. Whenever a vertex $w$ gets a message ``$v_0$ says \ldots says $v_t$ says $m$ to $v$’’, it sends all its neighbors other than the $v_i$s the message ``$w$ says $v_0$ says \ldots says $v_t$ says $m$ to $v$’’. When $v$ gets a message for it, it finds a set $F$ of (possibly-faulty) vertices such that deleting all messages containing them leaves a consistent message $m$, and reads that message. 

Such a set $F$ exists because the faulty vertices themselves are one such set, so it suffices to show that there aren’t two distinct sets of possibly-all-faulty $F_1$ and $F_2$ such that deleting all messages containing a vertex in $F_1$ gives one message $m_1$ and deleting all messages containing a vertex in $F_2$ gives a different message $m_2$. Suppose for contradiction that there are. Consider $G \setminus (F_1 \cup F_2)$; we claim it's connected. Suppose not; then there is a cut $C \subseteq F_1 \cup F_2$ in $G$. Then $(C \cap F_1) \cup (C \cap F_2)$ is a cut, so by the conditions of the theorem, either not all the vertices in $C \cap F_1$ can simultaneously be faulty or not all the vertices in $C \cap F_2$ can simultaneously be faulty, so either not all the vertices in $F_1$ can simultaneously be faulty or not all the vertices in $F_2$ can simultaneously be faulty, contradicting the choice of $F_1$ and $F_2$. Hence $G \setminus (F_1 \cup F_2)$ is connected, so it contains a path $P$ consisting of vertices $u = v_0$, $v_1$, \ldots, $v_t=v$. Along that path, $v$ receives the message ``$v_{t-1}$ says\ldots says $u$ says $m$ to $v$’’. But $G \setminus F_1$ contains $P$ and by assumption all the messages in $G \setminus F_1$ agree on $m_1$, so $m_1 = m$, and $G \setminus F_2$ contains $P$ and by assumption all the messages in $G \setminus F_2$ agree on $m_2$, so $m_2 = m$. Hence $m_1 = m_2$, contradiction. So $v$ can determine the unique message sent by $m$, as desired.
\end{proof}

Bellow we state an equivalent theorem statement, which may be more intuitive, but less descriptive. 

\begin{thm}
If the weak cut property holds for every minimal vertex cut in the graph then every non-faulty vertex $v$ can communicate a message to any non-faulty vertex $w$ such that the message is guaranteed to be the same contents as the message sent by $v$.
\end{thm}

\subsection{If the weak cut property holds agreement is possible}

Now, we show that given the pairwise communication protocol and a small enough proportion of faulty nodes, the non-faulty nodes can reach agreement. 

\begin{thm}
If the weak cut property holds for every minimal vertex cut in the graph and $3f < n$ then agreement is possible.
\end{thm}

\begin{proof}

We can simulate a fully connected graph of size $n$, if every node can communicate with every other node pairwise. 
On a fully connected graph of size $n$ agreement is possible when the number of faults is less than one third of the total number of nodes. This holds due to the condition that $3f < n$.  Thus, running the standard byzantine fault tolerance algorithm for a fully connected graph, but where each pairwise communication is replaced with our pairwise communication protocol, will reach agreement. 

\end{proof}

\subsection{Application to Majority Non-Faulty Nodes Per Cut}
If a cut is of size $c$ and there are $f$ faults where $c>2f$ then $|A| \geq \lceil c/2 \rceil$ or $|B| \geq \lceil c/2 \rceil$. Furthermore: $\lceil c/2 \rceil > f$. Thus, $|A|>f$ or $|B|>f$ thus at least one of $A$ and $B$ are large enough that they are guaranteed to have at least one non-faulty node.  This allows any two nodes to communicate messages to each-other.  Thus, if $n > 3f$ then agreement can be achieved. 

\subsection{Application to Unknown Network Graph}

The above theorems assume that vertices know the underlying network graph only to check that the nodes guessed faulty, $F_1$ or $F_2$, are a minority of each cut. One might worry that faulty nodes could claim an incorrect graph structure in which they are a majority of some cut. However, with the slight modification that each vertex also sends its set of neighbors as part of every message, the algorithm works anyway, and knowledge of the graph is unnecessary: again, if $F_1$ and $F_2$ are two distinct sets of possibly-all-faulty vertices, let $T$ be the connected component of $G \setminus (F_1 \cup F_2)$ containing $v$. All the vertices in $T$ tell $v$ their sets of neighbors, and $v$ trusts those messages since they're from vertices in neither $F_1$ nor $F_2$, so $B$ knows the set $\delta(T)$ (neighbors of vertices in $T$ that aren't themselves in $T$, which form a minimal vertex cut). Then $F_1 \cap \delta(T)$ and $F_2 \cap \delta(T)$ are two sets of vertices that can all simultaneously be faulty whose union is reliably known by $v$ to be a minimal vertex cut, contradiction.

Thus, if we have the requirement of theorem 2 and $n > 3f$ then agreement can be achieved. If $2f < c$ then the condition in theorem 2 is met. Thus, if $n >3f$ and $c >2f$ then agreement can be achieved in an unknown network graph. 

\subsection{Application to Nodes Non-Faulty by Assumption}
If a node is non-faulty by assumption and theorem 2 is met then agreement is possible. Any cut with a non-faulty node can not meet the condition in theorem 1. Thus, if all cuts not containing a non-faulty node meet the condition in theorem 2 then all cuts meet the condition of theorem 2. Thus, pair-wise communication between nodes is possible. Thus, all nodes can communicate with the non-faulty node. The non-faulty by assumption nodes can elect a leader. Then, this leader can broadcast to all nodes it's value. All non-faulty nodes can agree to this value and achieve agreement. A Byzantine node broadcasting a value will be ignored, as it isn't a non-faulty by assumption node. 

If a subgraph $G$ is guaranteed to have more than $1/2$ of nodes be non-faulty and for all non-faulty nodes to reach agreement and theorem 2 is met, we can achieve agreement. Every node can communicate pair-wise because theorem 2 is effect. First all nodes in the subgraph $G$ communicate and get agreement between themselves. Then all nodes in $G$ communicate the agreement to the rest of the graph. Each node not in $G$ takes a vote over the values broadcast from nodes in $G$ and agrees to that value. All nodes in $G$ agree to the agreement value. All non-faulty nodes will agree to the same agreement value.

\section{Conclusion}
We have given a strictly more general characterization of the connectivity requirement for solving agreement with Byzantine Faults. In particular, our algorithm solves the agreement problem when the number of faults is less than a third of the number of nodes and no minimal vertex-cut set has a majority faulty nodes. It is easily seen that the prior condition of the number of faults being strictly less than half the graph connectivity satisfies our result. This generalization, allows for solving problems where we have more information about the distribution of faulty nodes, but, have more faults than halve of the smallest cut. 

We further show that our algorithm solves agreement even when lacking a description of the network on which it is executing; something that couldn't be done with the original agreement algorithm. We also show that, if we have a trusted node and our connectivity condition, agreement can be solved. These are some examples of cases where our more flexible algorithm and condition can allow for practical assumptions to make strong statements about systems. In addition, these extensions may help us tackle other open questions like agreement in dynamic graphs and robust network design.

\bibliographystyle{plain}
\bibliography{PaperBib}
\end{document}